\documentclass[9pt,conference]{IEEEtran}

\begingroup
  \def\x{\endgroup\ExecuteOptions{dvipdfm}}%
  \ifx\pdfoutput\undefined
  \else
    \ifx\pdfoutput\relax
    \else
      \ifcase\pdfoutput
        
      \else
        \def\x{\endgroup\ExecuteOptions{pdftex}}%
	
      \fi
    \fi
  \fi
\x

\usepackage{tikz}  

\usepackage{graphicx}
\usepackage{amsfonts}
\usepackage{amsmath}
\usepackage{amssymb}
\usepackage{mathtools}
\usepackage{longtable}
\usepackage{ulem}
\usepackage{url}
\usepackage{ulem}
\usepackage[T1]{fontenc}
\usepackage{times}
\usepackage{graphicx}
\usepackage{amssymb}																
\usepackage{array}
\usepackage[boxed,vlined,linesnumbered,noend]{algorithm2e}
\usepackage{tikz}
\usetikzlibrary{matrix,arrows,decorations.pathmorphing, ,decorations.pathreplacing, calc, shapes,patterns, fit, positioning}
\usepackage{subfig} 
\usepackage{float}

\newtheorem{theorem}{Theorem}

\newtheorem{definition}{Definition}
\newtheorem{corollary}{Corollary}
\SetKwInOut{Input}{Input}
\SetKwInOut{RestInput}{}
\SetKwInOut{Output}{Output}
\SetKwInOut{RestOutput}{}
\SetKwInOut{Actors}{Actors}
\SetAlFnt{\small}



\newenvironment{proof}{\par\noindent {\it Proof.}}										
		{\begin{flushright} \vspace*{-1mm} \mbox{$\Box$} \end{flushright}}

\makeatletter
\newlength\min@xx
\newcommand*\xxrightarrow[1]{\begingroup
  \settowidth\min@xx{$\m@th\scriptstyle#1$}
  \@xxrightarrow}
\newcommand*\@xxrightarrow[2][]{
  \sbox8{$\m@th\scriptstyle#1$}  
  \ifdim\wd8>\min@xx \min@xx=\wd8 \fi
  \sbox8{$\m@th\scriptstyle#2$} 
  \ifdim\wd8>\min@xx \min@xx=\wd8 \fi
  \xrightarrow[{\mathmakebox[\min@xx]{\scriptstyle#1}}]
    {\mathmakebox[\min@xx]{\scriptstyle#2}}
  \endgroup}
\newcommand*\xxleftarrow[1]{\begingroup
  \settowidth\min@xx{$\m@th\scriptstyle#1$}
  \@xxleftarrow}
\newcommand*\@xxleftarrow[2][]{
  \sbox8{$\m@th\scriptstyle#1$}  
  \ifdim\wd8>\min@xx \min@xx=\wd8 \fi
  \sbox8{$\m@th\scriptstyle#2$} 
  \ifdim\wd8>\min@xx \min@xx=\wd8 \fi
  \xleftarrow[{\mathmakebox[\min@xx]{\scriptstyle#1}}]
    {\mathmakebox[\min@xx]{\scriptstyle#2}}
  \endgroup}

\usepackage{xcolor}

\newcommand{\Pra}{\mathbb{P}}
                            
\begin{document}

\title{How to Cooperate Locally to Improve Global Privacy in Social Networks? \\
 On Amplification of  Privacy Preserving Data Aggregation 
}

\author{\IEEEauthorblockN{Krzysztof Grining}
\IEEEauthorblockA{Department of Computer Science\\
Faculty of Fundamental \\ Problems of Technology, WUT}
\and
\IEEEauthorblockN{Marek Klonowski}
\IEEEauthorblockA{Department of Computer Science\\
Faculty of Fundamental \\ Problems of Technology, WUT}
\and
\IEEEauthorblockN{Malgorzata Sulkowska}
\IEEEauthorblockA{Department of Computer Science\\
Faculty of Fundamental \\ Problems of Technology, WUT}}


\maketitle\footnote{Supported by Polish National Science Center - NCN, decision number DEC-2013/ 08/M/ST6/00928 (Harmonia)}

\begin{abstract}
In many systems  privacy of users  depends on the number of participants applying collectively 
some method to protect their security. Indeed, there are numerous already classic results about revealing aggregated 
data from a set of users. The conclusion is usually as follows: if you have enough friends to ``aggregate'' the private data, 
you can safely reveal your private information. 

Apart from data aggregation,  it has been noticed that in a wider context privacy can be often reduced to being hidden in a crowd. 
Generally, the problems is how to create such crowd. This task may be not easy in some distributed systems, wherein gathering
enough ``individuals'' is hard for practical reasons.   

Such example are social networks (or similar systems), where users have only a limited number of semi trusted contacts and their aim is to 
reveal some aggregated data in a privacy preserving manner.  This may be particularly  problematic in the presence of a strong  adversary  that can 
additionally corrupt some users.

We show two   methods that allow to significantly amplify privacy with only limited number of local operations and very 
moderate communication overhead. Except theoretical analysis we show experimental results on topologies of  real-life social networks 
to demonstrate that our methods can significantly amplify privacy of chosen aggregation protocols even facing a massive attack of a 
powerful adversary.

We believe however  that our results can have much wider applications for improving security of systems based on locally trusted relations.

\end{abstract}

 \vspace*{-10pt}

\section{Introduction}\label{sect:intro}

Most algorithms providing anonymity or privacy in distributed systems consist in hiding an element in a group of other elements. 
Indeed, one of the very first  definitions of anonymity from ~\cite{ANODEF0} describes it as a \textit{
state of being not identifiable within a set of subjects, the  ``anonymity set''.}

Similar approach to privacy in the context of data bases is caught in  \textit{$k$-anonymity} metrics (~\cite{samarati2001protecting,d-sweeney2,d-sweeney3}). 
That is, the privacy is preserved as long as  each element is revealed in a group of at least  $k$ other, identical 
elements. In this metric as well as some consecutive concepts like $\ell$-\textit{diversity} \cite{d-machanava} or $m$-\textit{invariance}~\cite{d-xiao},  
the bigger the ``anonymity set'' is, the stronger the privacy guarantees are. This idea is also  reflected in further definitions 
of anonymity/privacy \cite{ANODEF2,ANODEF1}. 

It turns out however that similar phenomenon can be also observed in systems typically investigated from \textit{differential privacy} perspective. 
Let us remind that this privacy metric is in fact a standard one and was introduced in the seminal paper~\cite{dW1}. 
In the context of distributed system of somehow connected individuals we usually consider a problem where some function of data has to be
revealed preserving privacy of individuals. Many real life cases fall into this scenario. 
The most obvious example is privacy preserving data aggregation, wherein we need to reveal e.g. a sum of values of users protecting their privacy at the same time. 
Such aim can be realized using combination of cryptography and the common trick of adding random value, (\textit{a noise}),
to the aggregated data (see for example \cite{PaniShi} and \cite{Rastogi}). It turns however that the bigger the set of individual contributed to the sum, 
the less  noise has to be added to protect privacy of individuals.  Alternatively, having the same level of privacy one can reveal 
more exact statistics if they refer to a bigger set of individuals. 

In our paper we consider a distributed system that consists of nodes (individuals having some local, possibly sensitive, data) with connections constituting
a graph modeled as a preferential attachment process. This model is believed to be appropriate for a wide spectrum of real systems including social networks. 
We may assume that each individual has a very constrained knowledge about the network limited to some \textbf{(semi)trusted} neighbors.  
We formally prove that using a simple algorithm (adding some extra connections between nodes) one may protect privacy of nodes 
even in the presence of an adversary capable of corrupting significant number of  nodes. What is more important, our algorithm needs only some moderate number of \textbf{local} operations. In particular, the exact topology or even its exact size may remain unknown. Apart from rigid formal analysis we provide experimental results performed on data from real networks. 

In Section~\ref{sect:model} we present a general model with adversary. In Section~\ref{sect:prot} we present our distributed protocols which may be used to enhance privacy (e.g. in social networks). In Sections~\ref{sect:results} and~\ref{sect:exper} we provide analytic and experimental analysis of our protocols. In Section~\ref{conse} we show two examples of privacy amplification using our generic method. Then in Section~\ref{sect:related} we present some related papers and finally in Section~\ref{sect:conclusion} we conclude and outline some interesting problems for future work.


			

\vspace*{-3pt}

 \section{General Adversarial Model}\label{sect:model}

We consider a social network represented as a graph $G$. 
The nodes and edges of $G$ represent users and  friendship relation between pairs of them, respectively.
 
Our model can be directly used for other distributed systems, wherein a privacy preserving data computation problem is considered (e.g., a sensor network or a systems of smart meters) as long as some assumptions typical for social networks  about the network topology are fulfilled.  

The intuition behind our adversarial model is as follows. We assume that the Adversary can corrupt some of the users. Corruption gives the Adversary control over the node, yet we assume that he is an honest-but-curious type of Adversary. Namely, corrupted nodes follow the protocol, but they are trying to learn information about processed data and share all information they have with the Adversary.

The corruption of nodes can either be done in a random way 
or the Adversary can choose an arbitrary subset of nodes to corrupt, knowing the exact structure of the graph (say, he may attack the nodes with the highest degree). 
Note that random corruptions can model scenario when some  users install protecting software and others remain attack-prone. This case also covers 
the situation  with  unexpected failures, without an actual presence of Adversary. 

Note that the Adversary has full access to all information processed by corrupted nodes. 
As we show later, from the perspective of the Adversary  all connections incident to a corrupted node are removed from the graph $G$.

\begin{definition}
We will say that a graph is \textit{$\xi$-strong} if a subgraph induced by its honest nodes has largest connected component of size at least $\xi n$, where $n$ is the number of honest nodes.
\end{definition}

Clearly,  corruption of a significant number of nodes can dramatically decrease the $\xi$-strength. For prevention, we can enrich the graph by adding some edges between users, e.g.,  between some arbitrary user and a friend of his friend. For practical reasons all these operations need to be \textbf{local} (no global topology is known)  
and simple. 


We formally define \textit{Disconnection Game}  with Adversary $\mathcal{A}$ and a distributed protocol $\mathcal{P}$ as follows.  We have a network with underlying undirected graph $G = (V,E)$
This can either be a specific real network graph or e.g., a  randomly generated  scale-free graph. We define Disconnection Game, denoted by $\mathcal{DG}(G,\mathcal{A},\mathcal{P})$ in the following way:
\begin{enumerate}
\item $\mathcal{P}$: the set of edges $E$ is enriched by adding edges chosen  between  pairs of unconnected nodes. Rules of adding  edges depend on  specific game instantiation. 
This resulting graph is  denoted $G_P = (V,E \cup E_P)$, where $E_P$ is the set of edges added after $\mathcal{P}$ was applied. 

\item The Adversary chooses, according to restrictions in this game instantiation, a subset $C$ of nodes. The nodes belonging to $C$ are \textit{corrupted} and removed from the graph with their incident edges denoted by $E_C$. Note that the Adversary knows only the initial graph $G$. 

The resulting graph is denoted  $G_A = (V \setminus C,(E \cup E_P)  \setminus E_C )$.

We assume that $C$ does not depend on the set $E_{P}$.  This assumption reflects the assumption that the adversary does not know choices of uncorrupted nodes.

\item The outcome of the game is the fraction of nodes belonging to the biggest connected component in graph $G_A$.
\end{enumerate}

The model presented in this section is a problem of robustness of the network (see for example \cite{beygelzimer2005improving,zhang2015notion,callaway2000network,zhao2009enhancing}). It is, however, worth mentioning that, unlike previous papers in that field, we require that the enhancing protocol is done in a distributed way and without much knowledge of global topology of the graph. Moreover, we pick rather strong notion of robustness, namely the size of the largest connected component.

If the resulting structure is $\xi$-strong, it means that there exists a structure that is \textbf{not} controlled by the adversary that is connected  and contains 
at least $\xi\cdot n$ out of $n$ nodes. Intuitively, this allows to provide a common response secured  in such way that the adversary cannot observe separate inputs of nodes but the 
aggregated value of a large set of nodes. In Section~\ref{conse} we present references to particular protocols. 




 \vspace{-3pt}
\section{Security-enhancing protocols}\label{sect:prot}

We present two protocols  aimed at improving  $\xi$-strength of the network   and in consequence security of aggregation protocols. We prove their properties 
both in analytic (Sec.~\ref{sect:results}) and experimental (Sec.~\ref{sect:exper}) way for underlying graphs typical for social networks. 
\vspace{-3pt}
\subsection{$m$-Two Steps Friend Finder}


The person who wants to improve his chances of being in the big component 
asks his friend (chosen uniformly at random) to recommend him yet to another friend. Namely, our new friend is a former "friend of a friend" that is added to the list of connections 
(or just a separated contact  used for privacy-preserving actions).
This procedure is iterated  $m$ times, namely ask $m$ randomly chosen friends 
for recommendations. That would result in obtaining (at most) $m$ new friends. Note that sometimes it might happen that a specific "friend of a friend" will be recommended more than once. 

Formally, every node that wants to actively participate in the protocol performs a random walk of length $2$ starting from himself. Note that one could propose different length of the random walk, our choice of length $2$ is to minimize communication and keep the protocol as local as possible. 

Formally the 
 $m$-Two Steps Friend Finder ($m$-2SFF, for shortness) is presented as an Algorithm~\ref{LocalRandomWalk}.

\begin{algorithm}
\DontPrintSemicolon
\caption{$m$-2SFF}\label{LocalRandomWalk}
\ForEach{node $v$}{
      \For{ $m$ times}{
				1. Choose node $w$ uniformly at random from $N(v)$.\;
				2. Query $w$ to get id of its neighbor.\;
				3. Node $w$ chooses $u$ uniformly at random from $N(w)$ and sends its id to $v$.\;
				4. Create edge $(v,u)$.\;
      }
  }
\end{algorithm}

Note that $m$-2SFF can be performed by a node without any knowledge of the underlying graph, except its neighbors. 
Moreover, it can be done in a fully distributed manner, with $O(mn')$ messages  sent in the network, where $n'\leq n$ is the number of nodes participating in the protocol. 



\vspace{-3pt}

\subsection{$m$-Ask Fat For a Friend}

The approach in this protocol is substantially different. Here we want to rely on the preferential attachment properties of real networks.
In particular, we assume that there is a commonly known list of a few nodes with highest degrees. We will call them \textit{fat nodes}. In real life situation we might think that there are a few well-known and somewhat trusted parties in the distributed system. 


Existence of such fat nodes is typical for structures governed by  preferential attachment  model (a.k.a. "rich get richer"). Note that  there is a vast research in this kind of models and it turns out that complex, real life networks tend to exhibit such  properties.

$m$-Ask Fat For a Friend ($m$ -A3F)  goes as follows. Every node that wants to improve its chance to belong to the big component 
 has to choose uniformly at random one fat node from the common list  and ask for an address of one of its neighbors chosen at random. 
Formally, $m$ -A3F protocol is presented in Algorithm~\ref{AskRichNode} and~\ref{AskRichNode2}.

\begin{algorithm}
\DontPrintSemicolon
\caption{$m$ -A3F (code for a regular node)}\label{AskRichNode}
\ForEach{node $v$}{
      \For{ $m$ times}{
				1. Choose node $w$ at random from the common list of fat nodes .\;
				2. Query $w$ to get id of its neighbor.\;
				3. Node $w$ chooses $u$ uniformly at random from $N(w)$ and sends its id to $v$.\;
				4. Create edge $(v,u)$.\;
      }
  }
\end{algorithm}

\begin{algorithm}
\DontPrintSemicolon
\caption{$m$ -A3F (code for a fat node)}\label{AskRichNode2}
\ForEach{fat node $w$}{
	\If{queried by node $v$}{
			 Reply with $u$ chosen  at random from $N(w)$.\;				
		}
}
\end{algorithm}

Using a list of 'fat' nodes may be perceived as a bottleneck of the protocol, yet one should easily realize that in many real life cases the fat node has significantly more resources.
 Think about the case where the network is the WWW and 'fattest' nodes are e.g., Google, Yahoo or Facebook. Moreover a fat node does not participate in further communication. It just contacts two nodes so that they can establish an independent connection. 

In the next Sections we show that using fat nodes for finding friends substantially improves the immunity of the graph even facing a massive attack of the adversary. 



 \section{Analytic results}\label{sect:results}
In this Section we analyse a specific, most interesting case of our protocols in a general model. Other cases are also considered in the next Section. Let us analyse the $\log{n}$-A3F with Adversary knowing the topology of graph $G$ in advance thus attacking the nodes with the highest degree. We consider $G=(V,E)$ to be preferential attachment graph having some properties that can be met in real-life networks. One of such properties is existence (whp) of a group of vertices having high (in some sense) degrees. Their neighborhood covers whp the linear number of vertices from $V$.

Thus, let us assume throughout this Section the following. Let $W \subset V$ be the subset of vertices whose degrees vary from $an/\log{n}$ to $bn/\log{n}$ for some constants $a, b$, $|W| = C \log{n}$ for some constant $C$. $W$ is the set of the fat nodes from our protocol and, at the same time, the set of vertices that will be corrupted by Adversary. By $N_W$ we denote the neighborhood of $W$ without vertices from $W$, thus $N_W = \bigcup_{i=1}^{|W|} N(w_i) \setminus W$, where $N(w_i)$ is the neighborhood of $w_i$. We assume also that $|V\setminus (W \cup N_W)| = \alpha n$ for some constant $0 < \alpha < 1$. Let $V_{\alpha} = V\setminus (W \cup N_W)$. We will use the well known fact about the Erd{\"o}s-Reny{\'i} $G(n,p)$ model (see for example \cite{bollobas1998random}), namely that whenever $p\geq (1+\varepsilon)\log{n}/n$ for some $\varepsilon>0$, then whp $G(n,p)$ is connected.


First, let us consider the case in which all vertices want to participate in the $\log{n}$-A3F Protocol.

\begin{theorem}
\label{th_all}
If $C a <1-\alpha$ then after executing $\log{n}$-A3F for all vertices in $G=(V,E)$ we obtain $G_A = (V\setminus W, (E \cup E_P)\setminus E_C)$ which is whp \textit{$1$-strong}. (Recall that $E_P$ is the set of edges added during the protocol execution and $E_C$ is the set of edges incident to vertices from $W$.)
\end{theorem}
\begin{proof}
Note that the set of vertices of $G_A$ satisfies $V\setminus W = N_W \cup V_{\alpha}$  and $N_W$ and $V_{\alpha}$ are disjoint.
First, let us concentrate on the set $N_W$. Let $u,v \in N_W$. Let $i$ be such that $\{w_i, v\} \in E$. Let us estimate the probability that there exists an edge $\{u,v\}$ (denote this event by $[u \leftrightarrow v]$). Let $[u \rightarrow\ v]$ denote the event that $u$ established an edge $\{u,v\}$ during the protocol.
For some $\varepsilon>0$ and sufficiently big $n$ we get
\begin{equation}
\begin{split}
\label{l_bound_uv}
\Pra[ & u \rightarrow v] \geq 1-\left(1-\frac{1}{C \log{n}}\frac{1}{deg(w_i)}\right)^{\log{n}} \geq \\
& 1-\left(1-\frac{1}{C \log{n}}\frac{\log{n}}{a n}\right)^{\log{n}} \geq \\
& 1-e^{-\frac{\log{n}}{C a n}} \geq \frac{\log{n}}{C a n+\log{n}}\geq (1+\varepsilon)\log{(|N_W|)}/|N_W|.
\end{split}
\end{equation}
Note that $1/(C \log{n}~deg(w_i))$ is the lower bound for the probability that $v$ establishes an edge $\{v,u\}$ in a single step of the protocol. Indeed, $w_i$ does not need to be the only neighbor of $u$ in $W$. The second inequality follows from the bounds for $deg(w_i)$.
The third inequality follows from the fact that $(1-1/x)^x$ converges to $1/e$ from below for $x>0$, the fourth one from the fact that $e^x \geq 1+x$. The last inequality follows because $C a<1-\alpha$ and $|N_W| = (1-\alpha)n-C\log{n}$. Since each vertex creates new edges during the protocol independently from other vertices, we have $\Pra[ u \leftrightarrow v] = \Pra[u \rightarrow v] + \Pra[v \rightarrow u] - \Pra[u \rightarrow v]\Pra[v \rightarrow u]$. Of course, the lower bound (\ref{l_bound_uv}) is true also for $\Pra[v \leftrightarrow u]$ for all $u,v \in N_W$. We can think that the subgraph of $G$ induced on $N_W$ (denote it by $G(N_W)$) decomposes into Erdos-Renyi $G(N_W,p)$, where $p \geq (1+\varepsilon)\log{(|N_W|)}/|N_W|$, and some other random graph. Thus $G(N_W)$ will inherit some monotone properties of $G(N_W,p)$, among others, it will be connected whp. Since Adversary corrupts the nodes with the highest degrees, namely the whole set $W$, all the vertices from $N_W$ will stay in $G_A$. Thus we have proved the existance (whp) of a giant component (which contains $G(N_W)$) of size at least $|N_W| = (1-\alpha)n-C\log{n}$ in $G_A$.

Now, let us concentrate on the set $V_{\alpha}$. Let us estimate the probability that a vertex $v \in V_{\alpha}$ is not connected with $G(N_W)$ (denote this event by $[v \not \leftrightarrow G(N_W)]$). What needs to happen is that whenever the fat node sends to $v$ the id of $u$, $u$ needs to be a fat node as well. Since there are $C \log{n}$ fat nodes and their degrees are at least $a n/\log{n}$, we obtain
\[
\Pra[v \not \leftrightarrow G(N_W)] \leq \left(\frac{1}{C \log{n}} \frac{C \log{n}}{a n/\log{n}}\right)^{\log{n}} = \left(\frac{\log{n}}{a n }\right)^{\log{n}}.
\]
Vertices from $V_{\alpha}$ act during the protocol independently and the above probability is so small that we can simply estimate the probability that all vertices from $V_{\alpha}$ are connected with $G(N_W)$ (denote this event by $[V_\alpha \leftrightarrow G(N_W)]$) and show that it happens whp:
\[
\Pra[V_\alpha \leftrightarrow G(N_W)] \geq \left(1-\left(\frac{\log{n}}{a n }\right)^{\log{n}}\right)^{\alpha n} \xrightarrow {n \rightarrow \infty} 1.
\]
Thus whp $G_A$ is connected.
\end{proof}

The above theorem gave us a very strong result however its assumption about the number of vertices taking part in the protocol was also very strong. Now, let us discuss the following case: $\beta$ fraction of vertices from $V_{\alpha}$ and $\gamma$ fraction of vertices from $N_W$ take part in the protocol. (We don't care about vertices from $W$ because they are going to be corrupted and their incident edges will not appear in $G_A$ eventually).

\begin{theorem}
If $C a <1-\alpha$ and $C b > \gamma(1-\alpha)$ then after executing $\log{n}$-A3F for vertices as described above on $G=(V,E)$ we obtain $G_A = (V\setminus W, (E \cup E_P)\setminus E_C)$ which is whp \textit{$(1-(1-\beta)\alpha)$-strong}.
\end{theorem}
\begin{proof}
Let $\tilde{N}_W$ denote the set of vertices from $N_W$ which take part in the protocol ($|\tilde{N}_W| = \gamma|N_W|$). Even though the vertices from $N_W \setminus \tilde{N}_W$ do not take part in the protocol, they can be chosen as those to whom vertices from $\tilde{N}_W$ establish new edges. Let us estimate the probability that $v \in N_W \setminus \tilde{N}_W$ will not get connected to any vertex from $\tilde{N}_W$ during the execution of the protocol (denote this event by $[v \not\leftrightarrow G(\tilde{N}_W)]$). Let $i$ be such that $w_i$ and $v$ are neighbors in $G$. We have
\[
\begin{split}
\Pra[ v \not\leftrightarrow & G(\tilde{N}_W)] \leq \left(1-\frac{1}{C \log{n}}\frac{1}{deg(w_i)} \right)^{\gamma|N_W|\log{n}} \leq \\
& \left(1-\frac{1}{C b n} \right)^{\gamma|N_W|\log{n}} \leq e^{-(\gamma \log{n} |N_W|)/(C b n)} = \\
& n^{-\gamma(1-\alpha)/(C b)} n ^{\log{n}/(b n)}
\end{split}
\]
(compare \ref{l_bound_uv}).

Now, let us estimate the probability that all vertices from $(N_W \setminus \tilde{N}_W)$ are going to be connected with $G(\tilde{N}_W)$ (denote this event by $[(N_W \setminus \tilde{N}_W) \leftrightarrow G(\tilde{N}_W)]$). We get
\[
\begin{split}
\Pra[&(N_W \setminus \tilde{N}_W) \leftrightarrow G(\tilde{N}_W)]] \geq \\
& \left(1-n^{-\gamma(1-\alpha)/(C b)} n ^{\log{n}/(b n)} \right)^{(1-\gamma)|N_W|} = \\
& \left(1-n^{-\gamma(1-\alpha)/(C b)} n ^{\log{n}/(b n)} \right)^{(1-\gamma)((1-\alpha)n - C \log{n})} \\
& \xrightarrow{n \rightarrow \infty} 1
\end{split}
\]
since $C b > \gamma(1-\alpha)$. Thus again whp $G(N_W)$ is connected.

By calculations analogous to those from Theorem \ref{th_all} we also get that all vertices from $V_{\alpha}$ which participate in the protocol (denote this set by $\tilde{V_{\alpha}}$) are connected with $G(N_W)$ whp. We proved that whp $G_A$ has a giant component containing $N_W \cup \tilde{V_{\alpha}}$, $|N_W \cup \tilde{V_{\alpha}}| = (1-\alpha)n - C \log{n} + \beta\alpha n$. This completes the proof.

\end{proof}
 \section{Experimental results}~\label{sect:exper}

We present experimental results conducted on \textbf{real} data of Epinions social network collected in SNAP dataset by Stanford University (see ~\cite{snapnets} and~\cite{richardson2003trust}). 

This is a who-trust-whom online social network of a a general consumer review site Epinions.com. Members of the site can decide whether to ''trust'' each other. All the trust relationships interact and form the Web of Trust which is then combined with review ratings to determine which reviews are shown to the user. Our network has $75879$ nodes and $508837$ edges where nodes denote users of Epinions.com site and edges denote trust relation.



\subsection{Random Failures}

Random failures is a widely used model in network robustness but also fault tolerance (see~\cite{Hubercik}) literature. 
We assume that corrupted nodes (or in other words, nodes which are prone to failure) are distributed in a uniform way across the whole network.




\subsubsection{$m$-2SFF Protocol}

First let us concentrate on the $m$-2SFF Protocol in the case of Random Failures. Initially we assume that all nodes launch  the $m$-2SFF Protocol, namely each node does $m$ random walks of length $2$ to establish extra connections.  Obviously, the larger $m$, the better safety of the nodes.

In Figure~\ref{fig:localRandom} we show how the $m$-2SFF Protocol performs on Epinions social network graph under Random Failures model. We can see how the network behaves without any enrichment, and with $m = 1,5,10,15$. Note that on the x-axis we have the percentage of corrupted nodes. With $m=15$ walks, around $70\%$ of remaining nodes are in the single  \textbf{giant} connected component. Note that the edges are added \textbf{before} the corruption phase. Therefore, for each remaining node, a lot of added neighbors are corrupted and therefore useless. On the positive side, one can easily see that for up to around $20\%$ failures, even $5$ walks are sufficient to have almost every node belonging to the giant component.

\begin{figure}[h!]
    \centering
    \includegraphics[width=0.5\textwidth]{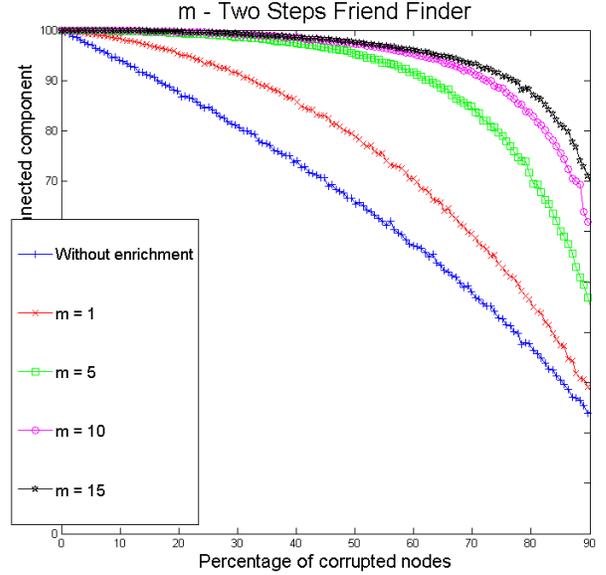}
    \caption{$m$-2SFF under Random Failures model.}
    {\label{fig:localRandom}}
\end{figure}

Despite these somewhat optimistic results, it is quite unrealistic to assume that all users want to participate. 
We want to weaken this assumption. 
We still demand high level of security, at least for the participating users. In the Figure~\ref{fig:partialLocalRandom} we show some experimental results when a part of nodes participates, only. Here we assume $m=15$ and $q=0.1, 0.25, 0.5$ fraction of participating nodes. That is,  $q\cdot n $ nodes  participate in $15$-2SFF protocol. Then we are interested what is the fraction of participating users that belong to the biggest  component and how it compares to the situation when all users do participate.

\begin{figure}[h!]
    \centering
    \includegraphics[width=0.5\textwidth]{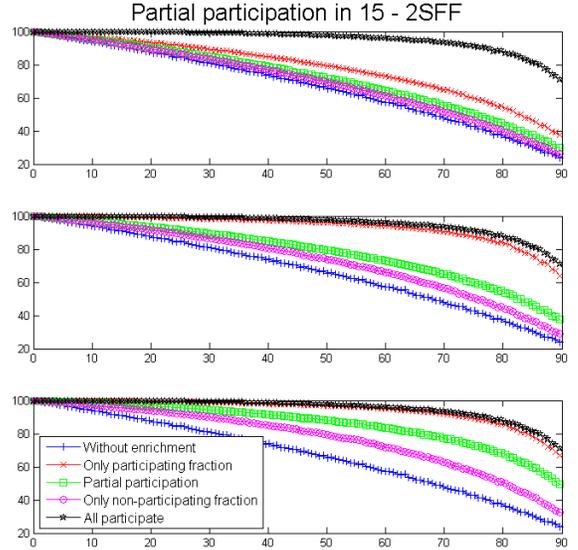}
    \caption{$15$-2SFF under Partial Participation and Random Failures model. The top figure shows $10\%$ participation, middle shows $25\%$ participation and bottom $50\%$ participation.}
    {\label{fig:partialLocalRandom}}
\end{figure}

Note that in the case where $q=0.1$ there is a significant decrease of security. Namely, with massive number of failures, we have around $30\%$ nodes in biggest  component in comparison to $70\%$ in the full participation case. Note that even if we consider only the subset of participating nodes, then the fraction of nodes belonging to biggest  component amongst them is below $40\%$. The security indeed improves with greater $q$, yet still even if we consider only the participating nodes, the results are significantly worse than when all users participate.
Thus this protocol turned to be useful in communities if we know that strong majority  of nodes is willing to use it. 

\subsubsection{$m$- A3F Protocol}

Now we focus on the $m$-A3F protocol  under Random Failures model. Again, we initially assume that all nodes participate in the protocol, namely each node does $m$ queries which consist of randomly choosing one of the fat nodes  and asking for randomly chosen neighbor of that node. 
Here we fixed the number of the nodes considered fat for $\left\lfloor\log(n)\right\rfloor = 16$. It means that 16 nodes which have the highest degree in the initial graph are 
on the common list of 'fat nodes`.

In Figure~\ref{fig:influentialNode} one can see the performance of A3F on Epinions social network graph under Random Failures model. Similarly as before, we show the behavior of the network without any enrichment, and with $m = 1,5,10,15$. This time, with $m=15$ queries, almost $90\%$ of remaining nodes are in the giant component despite of a large  number of failures. Another interesting thing to observe is that the cutoff (moment when the fraction of nodes in the giant component begins to decrease significantly) appears much farther. For example, in case of $15$-2SFF we see that the size of the giant component starts to deteriorate since approximately $30\%$ failures, before this threshold it remains very  close to $100\%$. In the case of 2SFF, on the other hand, for $m=15$ the cutoff appears as far as $70\%$ failures and before such a massive corruption of nodes, it remains negligibly close to $100\%$.

\begin{figure}[h!]
    \centering
    \includegraphics[width=0.5\textwidth]{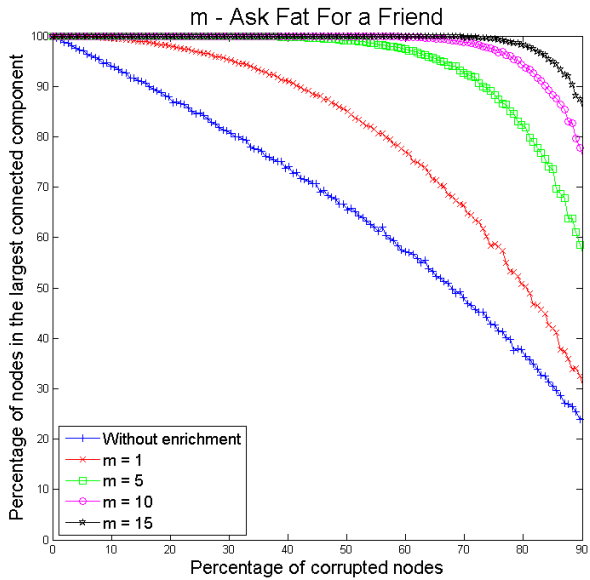}
    \caption{$m$-A3F under Random Failures model.}
    {\label{fig:influentialNode}}
\end{figure}

Again we are interested in the performance of A3F in the case where only a fraction of users wants to participate. We assumed $m=15$ and $q=0.1, 0.25, 0.5$ participation. In Figure~\ref{fig:partialInfluentialNode} we have shown the results for 2SFF with partial participation. 

\begin{figure}[h!]
    \centering
    \includegraphics[width=0.5\textwidth]{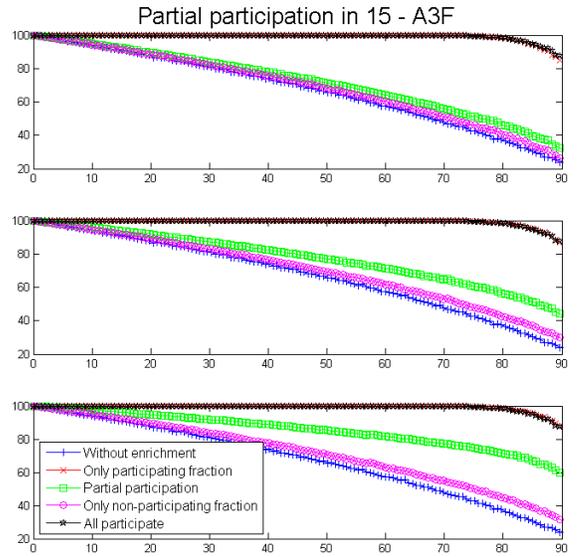}
    \caption{$15$-A3F under Partial Participation and Random Failures model. The top figure shows $10\%$ participation, middle shows $25\%$ participation and bottom $50\%$ participation.}
    {\label{fig:partialInfluentialNode}}
\end{figure}

The most interesting thing is the fact that the safety level amongst the participating nodes in case of partial participation is virtually the same as the safety level when all nodes participate. This fact is very important from the practical point of view. It gives the users a choice - whether they want to sacrifice their safety and not participate in the protocol, or participate in the protocol and be safe no matter what other users choose as long as at least some fraction (say $10\%$) decides to participate in the protocol.

\subsubsection{Comparison}
A glance at the figures in this subsection is enough to see that $m$-A3F performs better than $m$-2SF under Random Failures regime. See for example that for $m=10$ and for $90\%$ failures the A3F protocol gives approximately $80\%$ nodes belonging to the giant component, while 2SFF gives only $60\%$. Moreover, for  the cutoff and therefore non-negligible deterioration of the fraction of nodes in the biggest component appears for greater fraction of failures than in $m$-2SF protocol.

Intuitively, these differences in the results stem from the fact that in A3F we leverage naturally emerging preferential attachment models in real, complex networks, while 2SFF does not really utilize this fact. Connecting to neighbors of fixed, high-degree set of nodes massively improves robustness of real networks.

\subsection{Targeted Adversary}

In this subsection we present experiments conducted under far stronger Adversary
that can corrupt nodes of the highest degree. Namely, if the Adversary has to corrupt $k$ nodes, she sorts the list of nodes by degree and corrupts first $k$ of them. 

Note that the Adversary only has access to the initial graph, without enrichment. 
Obviously, for a specific instance of the graph one could possibly devise a more clever way of attack, 
however this strategy seems to be optimal in general. 
Note that complex network which resemble preferential attachment features are extremely prone to such attacks. 


\subsubsection{$m$-2SFF protocol}


In Figure~\ref{fig:targetedAttackLocalRandom} we show how $m$-2SFF performs on Epinions social network graph under Targeted Adversary model. We can see how the network behaves without any enrichment, and with $m = 1,5,10,15$. Note that on the x-axis we have the percentage of corrupted nodes and this time it ranges from $0$ to $30\%$ instead of $0-90\%$ due to the Adversary's strength. Note that without enrichment the fraction of nodes in the biggest\footnote{Note that from graph-theory perspective we have in this case a \textit{giant component} - a single component that contains a fraction of all nodes} component dramatically falls to almost $0$ for $20\%$ failures. In other words, if the Adversary destroys $20\%$ nodes of highest degree, the remaining graph consists only of very small components. On the other hand, see that for up to $5\%$ corruptions the $m=15$ walks version gives almost $100\%$ nodes belonging to the biggest component. Even for $30\%$ corruption the fraction of nodes in the biggest component is considerably large (approximately $60\%$). Recall that without enrichment under such a strong adversary there is virtually no giant component whatsoever.

\begin{figure}[h!]
    \centering
    \includegraphics[width=0.5\textwidth]{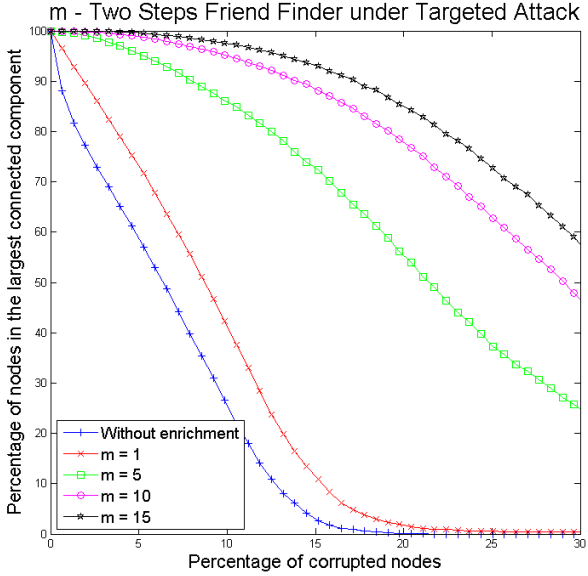}
    \caption{$m$-2SFF under Targeted Attack.}
    {\label{fig:targetedAttackLocalRandom}}
\end{figure}

Let us investigate the protocol if we assume that only a fraction of non-corrupted users participate actively. We assumed $m=15$ and $q=0.1, 0.25, 0.5$ participation. In Figure~\ref{fig:partialTargetedAttackLocalRandom} we have shown the results for $m$-2SFF with partial participation under Targeted Adversary regime. 

\begin{figure}[h!]
    \centering
    \includegraphics[width=0.5\textwidth]{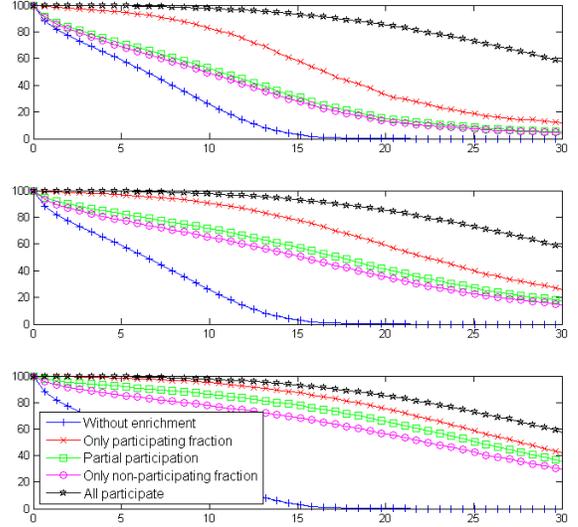}
    \caption{Partial $15$-2SFF under Targeted Attack. The top figure shows $10\%$ participation, middle shows $25\%$ participation and bottom $50\%$ participation.}
    {\label{fig:partialTargetedAttackLocalRandom}}
\end{figure}

An interesting difference between the results for this model and Random Failures can be seen in this figure. Namely, the fraction of nodes belonging to the giant component amongst those who participate is only slightly greater than amongst those who do not participate. This is highly undesired, as it gives no notion of improvement and benefit of participating actively in the protocol. A node could decide that it is pointless to waste precious resources and rather hope that the others would participate actively. See that even if half of the users actively participate, the fraction of nodes in the giant component are significantly smaller than when all nodes participate.

\subsubsection{$m$-A3F protocol}

After somewhat unsatisfying results for $m$-2SFF under Targeted Adversary, we will now present experiments on the $m$-A3F protocol. As before, let us first assume that all nodes participate in the protocol.

In Figure~\ref{fig:targetedAttackInfluentialNode} one can see the performance of $m$-A3F on Epinions social network graph under Targeted Adversary model. As previously, we show the behavior of the network without any enrichment, and for the cases where $m = 1,5,10,15$. This time, with $m=15$ queries, approximately $85\%$ of remaining nodes are in the biggest component for up to $30\%$ corruptions and over $95\%$ of nodes are in the giant component for up to $15\%$ corruptions. Another interesting thing to observe is that the cutoff again appears for greater number of corruptions. For example, in case of $15$-2SFF we see that the size of the giant component starts to deteriorate since approximately $5\%$ failures, before this threshold it remains close to $100\%$. In the case of $15$-A3F, on the other hand, the cutoff appears as far as at $10\%$ failures.

\begin{figure}[h!]
    \centering
    \includegraphics[width=0.5\textwidth]{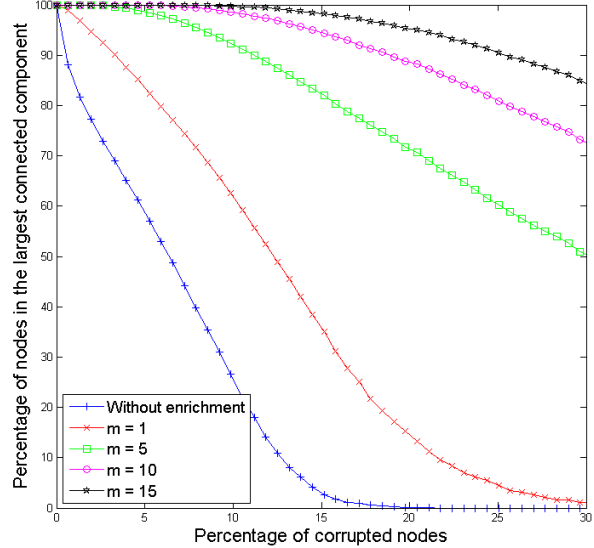}
    \caption{$m$-A3F under Targeted Attack.}
    {\label{fig:targetedAttackInfluentialNode}}
\end{figure}

Similarly as in the previous subsection, we want to see how the protocol behaves if we assume that only a fraction of non-corrupted users participate actively. We assumed $m=15$ and $q=0.1, 0.25, 0.5$ participation. In Figure~\ref{fig:partialTargetedAttackInfluentialNode} we show the results for A3F with partial participation under Targeted Adversary regime. 

\begin{figure}[h!]
    \centering
    \includegraphics[width=0.5\textwidth]{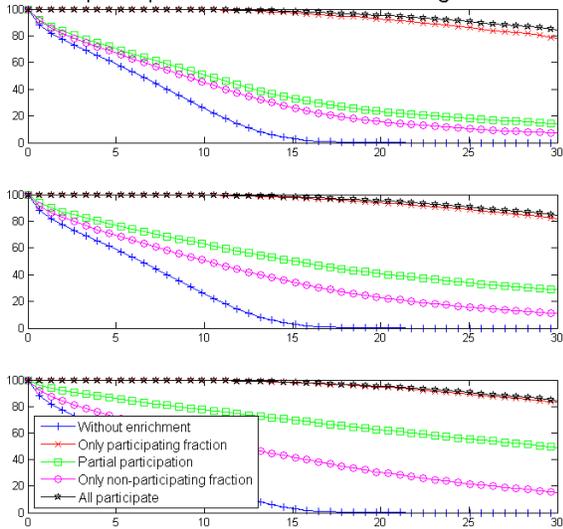}
    \caption{Partial $15$-A3F under Targeted Attack. The top figure shows $10\%$ participation, middle shows $25\%$ participation and bottom $50\%$ participation.}
    {\label{fig:partialTargetedAttackInfluentialNode}}
\end{figure}

Figure~\ref{fig:partialTargetedAttackInfluentialNode} is probably the most striking one due to the fact that in all three cases, one can easily see that the fraction of nodes belonging to the giant component amongst the actively participating nodes is almost the same as when all nodes participate. This is a \textbf{very desirable} feature of $m$-A3F because it gives the user a natural choice - participate in the protocol, which costs come computational resources, but be in the giant component independently of the choices of other nodes or do not participate, but then you are facing serious risk of ending up disconnected from the giant component.

\subsubsection{Comparison}

First of all, the results for both protocols are obviously worse than for Random Failure model, which is not surprising. However, they still give a significant improvement of the size of the giant component. Moreover, in the regime of Targeted Adversary, the $m$-A3F has a very interesting property of assuring almost the same fraction of nodes belonging to the giant component for participating fraction of nodes (even if only $10\%$ of users participate) as in the case where all users participate. 

This regime shows that $m$-A3F is indeed a very powerful enrichment to the graph structure. Note that we went from no giant component for $20\%$ failures to almost $90\%$ nodes belonging to the giant component amongst the actively participating nodes even if only $10\%$ of users participate. This scenario shows a significant improvement of security which is gained via $m$-A3F for those who actively participate in it. Note that the difference between the performance of $m$-2SFF and $m$-A3F is strongly connected with utilizing preferential attachment in real networks.

\section{Some Consequences for Security and Privacy}\label{conse}

Let us assume that graph $G$ is $\xi$-strong.

\begin{corollary}
Assume that we have a network with underlying graph $G$ which is $\xi$-strong. Then using cryptographic methods for data aggregation (see for example~\cite{PaniShi}) one can aggregate data even without adding noise. Such results are already presented in literature and require appropriate amount of users participating (see \cite{bassily2013coupled,bhaskar2011noiseless, kifer2012rigorous})
\end{corollary}

\begin{corollary}
Assume that we have a network with underlying graph $G$ which is $\xi$-strong. Then aggregation protocol PAALEC from~\cite{NaszeACNS} with parameters $\alpha = \exp(\frac{\epsilon}{\Delta})$ and $\beta = \frac{2\log(1/\delta)}{s}$ applied to graph $G$ is $(\epsilon,\delta)$-differentially private for the nodes belonging to the largest connected component. Moreover, PAALEC aggregation protocol is $(\epsilon,\delta + (1-\xi))$-differentially private for any arbitrary node.
\end{corollary}
 \section{Previous and Related Work}\label{sect:related}

This paper spans several areas, thus many different papers should be pointed as related work. Since the idea of scale free network modeling appeared, there has been a vast amount of research concerning these kind of networks, including classic papers like \cite{albert2002statistical, barabasi2003scale, kumar2000stochastic, aiello2000random, strogatz2001exploring}. Also worth mentioning are papers which provided rigorous mathematical treatment for scale free networks \cite{bollobas2001degree, bollobas2003mathematical, bollobas2004diameter}. More recent papers on properties of scale free networks include \cite{barabasi2009scale, fotouhi2013degree}. Also worth mentioning are papers \cite{chung2002connected,chung2004average} where authors consider various properties of a graph given its expected degree list.

We should also mention papers about community structure in large networks \cite{leskovec2009community,van2009random}.
Some empirical result can also be found in \cite{clauset2009power}.

The problem of robustness in complex networks has also been widely analyzed. To mention a few papers concerning the robustness and enhancing of robustness in scale free networks we cite \cite{zhao2009enhancing,tanaka2012dynamical, zhang2015notion,yang2015improving}. 
One should also mention~\cite{flaxman2005adversarial} wherein authors consider adversarial deletion in scale free graphs and \cite{beygelzimer2005improving}, where authors improve graph robustness by edge modifications.
Note that, in the network robustness literature the notion of robustness is mostly the fact that the largest connected component exists. Here, however, we are interested in non-asymptotic results and more precise size (or lower bound for the size) of the giant component. Moreover, our protocols can be performed locally and without knowledge of the graph topology.

Furthermore, papers concerning various anonymity and 'crowd-blending' concepts should be mentioned. See for example \cite{ANODEF0, ANODEF1, ANODEF2, d-sweeney2, d-sweeney3} for $k$-anonymity. See also \cite{d-machanava,d-xiao} for extensions and variations of anonymity.

We should also mention some privacy preserving papers with emphasis on those which could benefit from having large connected component of appropriate size, namely \cite{bassily2013coupled,bhaskar2011noiseless, kifer2012rigorous, NaszeACNS}. Also important are the papers \cite{PaniShi, Hubercik} where authors use cryptographic methods to amplify privacy for large group of users in data aggregation scenario. For survey about privacy see \cite{DworkAlgo} and references therein.

 \section{Conclusions and Future Work}\label{sect:conclusion}

We presented how to improve the size of the largest connected component under massive adversarial attack and demonstrated 
why this observation is important for a wide range of applications (with most emphasis put on privacy preserving protocols). Moreover, our methods are conceptually simple and can be performed locally, i.e. with minimal knowledge about the global network. We proved that the presented methods are efficient in preferential-attachment graphs, which are commonly believed to be an accurate model of various real-life networks including social interaction networks, World Wide Web, airline networks and many other. Finally, we confirmed our observations using experiments on graphs of \textbf{real} networks. 

We believe that many  questions  important both for theory as well as design of practical privacy preserving solutions are left unanswered. In particular, for future work we plan to investigate:  
\begin{itemize}
 \item even stronger Adversary, who can choose adaptively (namely during the enhancement protocol) vertices to corrupt;
 \item longer random walks, where we establish an edge with every node visited on the way;
 \item Our protocols improve security of participating individuals, but the level of privacy is improved also for other users. The questions is, how to design a mechanism (i.e., via constructing extra incentives) to improve global privacy dependently on a power of the adversary. 
\end{itemize}

%
%
%
 \bibliographystyle{IEEEtran}

\bibliography{bibliography} 

\end{document}